\documentclass[12pt]{article}
\usepackage[latin1]{inputenc}
\usepackage{geometry}
\usepackage{amsfonts}
\usepackage{amsmath}
\usepackage{amssymb}
\usepackage{amsthm}
\usepackage{graphics}

\usepackage{hyperref}

\newcommand{\eps}{\varepsilon}
\newcommand{\Q}{\mathbb{Q}}
\newcommand{\QQ}{\tilde{\mathbb{Q}}}

\newcommand{\F}{\mathbb{F}}
\newcommand{\R}{\mathbb{R}}

\newcommand{\cut}{\cap}
\newcommand{\uni}{\cup}

\newcommand{\prc}[1]{\ensuremath{\mathsf{#1}}}

\newcommand{\RP}{\prc{RP}}
\newcommand{\NP}{\prc{NP}}
\newcommand{\FPRAS}{FPRAS}

\newcommand{\SP}{\prc{\#P}}

\newcommand{\poly}{\mathsf{poly}}

\newcommand{\rk}{rk}

\newcommand{\gv}[1]{\mathbf{#1}}  
\newcommand{\vc}[1]{\mathbf{#1}}  
\newcommand{\Ga}{G_{aa}}
\newcommand{\ak}{a_1, \ldots, a_k}
\newcommand{\sak}{\{ \ak \} }
\newcommand{\red}{\preceq_T}
\newcommand{\redm}{\preceq^\mathtt{m}}

\newcommand{\aA}{\mathcal{A}}
\newcommand{\II}{\tilde{I}}
\newcommand{\cc}{\tilde{c}}
\newcommand{\NN}{\tilde{N}}

\newcommand{\subs}{\subseteq}
\newcommand{\y}{\upsilon}
\newcommand{\x}{\xi}
\renewcommand{\u}{\mu}

\newtheorem{defi}{Definition}[section]
\newtheorem{cor}[defi]{Corollary}
\newtheorem{lem}[defi]{Lemma}
\newtheorem{thm}[defi]{Theorem}
\newtheorem{rem}[defi]{Remark}
\newtheorem{prop}[defi]{Proposition}

\numberwithin{equation}{section}

\begin{document}

\title{On the Complexity of the Interlace Polynomial%
  \footnote{A preliminary version of this work has appeared in the
    proceedings of STACS 2008.}}

\author{Markus Bl\"aser, Christian Hoffmann}

\maketitle

\begin{abstract}
  We consider the two-variable interlace polynomial introduced by
  Arratia, Bollobás and Sorkin (2004). We develop graph
  transformations which allow us to derive point-to-point reductions
  for the interlace polynomial.  Exploiting these reductions we obtain
  new results concerning the computational complexity of evaluating
  the interlace polynomial at a fixed point. Regarding \emph{exact}
  evaluation, we prove that the interlace polynomial is \SP-hard to
  evaluate at every point of the plane, except on one line, where it
  is trivially polynomial time computable, and four lines, where the
  complexity is still open. This solves a problem posed by Arratia,
  Bollobás and Sorkin (2004). In particular, three specializations of
  the two-variable interlace polynomial, the vertex-nullity interlace
  polynomial, the vertex-rank interlace polynomial and the independent
  set polynomial, are almost everywhere \SP-hard to evaluate, too. For
  the independent set polynomial, our reductions allow us to prove
  that it is even hard to \emph{approximate} at every point except
  at~$0$.
\end{abstract}

\section{Introduction}

The number of Euler circuits in specific graphs and their interlacings
turned out to be a central issue in the solution of a problem related
to DNA sequencing by hybridization \cite{dna_interlacings}. This led
to the definition of a new graph polynomial, the one-variable
interlace polynomial \cite{interlace_polynomial}. Further research on
this polynomial inspired the definition of a two-variable interlace
polynomial $q(G;x,y)$ containing as special cases the following graph
polynomials: $q_N(G;y)=q(G;2,y)$ is the original one-variable
interlace polynomial which was renamed to ``vertex-nullity interlace
polynomial'', $q_R(G;x)=q(G;x,2)$ is the new ``vertex-rank interlace
polynomial'' and $I(G;x)=q(G;1,1+x)$ is the independent set
polynomial\footnote{The independent set polynomial of a graph $G$ is
  defined as $I(G;x)=\sum_{j\geq 0} i(G;j)x^j$, where $i(G;j)$ denotes
  the number of independent sets of cardinality $j$ of $G$.}
\cite{arratia_two_var_interl}.

Although the interlace polynomial $q(G;x,y)$ is a different object
from the celebrated Tutte polynomial (also known as dichromatic
polynomial, see, for instance, \cite{tutte_graph_theory}), they are
also similar to each other.  While the Tutte polynomial can be defined
recursively by a deletion-contraction identity on edges, the interlace
polynomial satisfies recurrence relations involving several operations
on vertices (deletion, pivotization, complementation).

Besides the deletion-contraction identity, the so called state
expansion is a well-known way to define the Tutte polynomial. Here the
similarity to the two-variable interlace polynomial is especially
striking: while the interlace polynomial is defined as a sum over all
vertex subsets of the graph using the rank of adjacency matrices (see
\eqref{eq:q_def}), the state expansion of the Tutte polynomial can be
interpreted as a sum over all edge subsets of the graph using the rank
of incidence matrices (see \eqref{eq:t_def})
\cite[Section~1]{arratia_two_var_interl}.

References to further work on the interlace polynomial can be found in
\cite{arratia_two_var_interl} and \cite{ems_distance_hereditary}. 

\subsection{Previous work}

The aim of this paper is to explore the computational complexity of
evaluating\footnote{See Section \ref{ssec:graph_polynomials} for a
  precise definition.} the two-variable interlace polynomial
$q(G;x,y)$. For the Tutte polynomial this problem was solved in
\cite{jaeger_vertigan_welsh}: Evaluating the Tutte polynomial is
\SP-hard at any algebraical point of the plane, except on the
hyperbola $(x-1)(y-1)=1$ and at a few special points, where the Tutte
polynomial can be evaluated in polynomial time. For the two-variable
interlace polynomial $q(G;x,y)$, only on a one-dimensional subset of
the plane (on the lines $x=2$ and $x=1$) some results about the
evaluation complexity are known.

A connection between the vertex-nullity interlace polynomial and the
Tutte polynomial of planar graphs \cite[End of Section
7]{interlace_polynomial}, \cite[Theorem 3.1]{ems_distance_hereditary} shows
that evaluating $q$ is \SP-hard almost everywhere on the line $x=2$
(Corollary~\ref{cor:qN_complexity}).

It has also been noticed that $q(G;1,2)$ evaluates to the number of
independent sets of $G$ \cite[Section 5]{arratia_two_var_interl},
which is \SP-hard to compute \cite{counting_IS_hard}. Recent work on
the matching generating polynomial \cite{averbouch_makowsky} implies
that evaluating $q$ is \SP-hard almost everywhere on the line $x=1$
(Corollary~\ref{cor:I_complexity}).

A key ingredient of \cite{jaeger_vertigan_welsh} is to apply graph
transformations known as stretching and thickening of edges. For the
Tutte polynomial, these graph transformations allow us to reduce the
evaluation at one point to the evaluation at another point. For the
interlace polynomial no such graph transformations have been given so
far.

\subsection{Our results}

We develop three graph transformations which are useful for the
interlace polynomial: cloning of vertices and adding combs or cycles
to the vertices. Applying these transformations allows us to reduce
the evaluation of the interlace polynomial at some point to the
evaluation of it at another point, see Theorem~\ref{thm:p2p},
Theorem~\ref{thm:p2p_comb} and Theorem~\ref{thm:p2p_cycles}. We
exploit this to obtain the following new results about the
computational complexity of $q(G;x,y)$.

We prove that the two-variable interlace polynomial $q(G;x,y)$ is
\SP-hard to evaluate at almost every point of the plane,
Theorem~\ref{thm:q_summary}, see also Figure~\ref{fig:pqcomplexity}.
Even though there are some unknown (gray, in
Figure~\ref{fig:pqcomplexity}) lines left on the complexity map for
$q$, this solves a challenge posed in \cite[Section
5]{arratia_two_var_interl}. In particular we obtain the new result
that evaluating the vertex-rank interlace polynomial $q_R(G;x)$ is
\SP-hard at almost every point (Corollary~\ref{cor:qR_complexity}).
Our techniques also give a new proof that the independent set
polynomial is \SP-hard to evaluate almost everywhere
(Corollary~\ref{cor:I_complexity}).

Apart from these results on the computational complexity of evaluating
the interlace polynomial \emph{exactly}, we also show that the values
of the independent set polynomial (which is the interlace polynomial
$q(G;x,y)$ on the line $x=1$) are hard to \emph{approximate} almost
everywhere (Theorem~\ref{thm:I_inapprox}).

\section{Preliminaries}
\label{sec:prelim}

\subsection{Interlace Polynomials}
\label{sec:interlace_polynomials}

We consider undirected graphs without multiple edges but with self
loops allowed. Let $G=(V,E)$ be such a graph and $A\subs V$. By $G[A]$
we denote $(A, \{e | e\in E, e \subs A \})$, the subgraph of $G$
induced by $A$. The adjacency matrix of $G$ is the symmetric $n\times
n$-matrix $M=(m_{ij})$ over $\F_2=\{0,1\}$ with $m_{i,j} = 1$ iff
$\{i,j\} \in E$. The rank of this matrix is its rank over $\F_2$.
Slightly abusing notation we write $\rk(G)$ for this rank. This allows
us to define the two-variable interlace polynomial.

\begin{defi}[{\cite{arratia_two_var_interl}}]
\label{def:q}
Let $G=(V,E)$ be an undirected graph. The interlace polynomial
$q(G;x,y)$ of $G$ is defined as
\begin{equation}
\label{eq:q_def}
q(G;x,y) = \sum_{A\subs V} (x-1)^{\rk(G[A])}
  (y-1)^{|A|-\rk(G[A])}.
\end{equation}
\end{defi}

In Section~\ref{sec:graph_trans} we will introduce graph
transformations which perform one and the same operation (cloning one
single vertex, adding a comb or a cycle to one single vertex, resp.)\
on every vertex of a graph.  Instead of relating the interlace
polynomial of the original graph directly to the interlace polynomial
of the transformed graph, we will analyze how, say, cloning \emph{one
  single vertex} changes the interlace polynomial. To express this, we
must be able to treat the vertex being cloned in a particular way,
differently from the other vertices. This becomes possible using a
\emph{multivariate} version of the interlace polynomial, in which each
vertex has its own variable.  Once we can express the effect of
cloning \emph{one} vertex by an appropriate substitution of the vertex
variable in the multivariate interlace polynomial, cloning \emph{all}
the vertices amounts to a simple substitution of all vertex variables
and brings us back to a bivariate interlace polynomial. This procedure
has been applied successfully to the Tutte polynomial
\cite{sokal-2005, blaes_mak}.

We choose the following multivariate interlace polynomial, which is
similar to the multivariate Tutte polynomial of Sokal
\cite{sokal-2005} and a specialization of the multivariate interlace
polynomial defined by Courcelle \cite{courcelle_interl}.
\begin{defi}
  Let $G=(V,E)$ be an undirected graph. For each $v\in V$ let $x_v$ be
  an indeterminate. Writing $x_A$ for $\prod_{v\in A}x_v$, we define
  the following multivariate interlace polynomial:
  \[ P(G;u,\gv{x}) = \sum_{A\subs V} x_A u^{\rk(G[A])}. \]
Substituting each $x_v$ in $P(G;u,\gv{x})$ by $x$, we obtain another
bivariate interlace polynomial:
\[ P(G;u,x) = \sum_{A\subs V} x^{|A|} u^{\rk(G[A])}. \]
\end{defi}
An easy calculation proves that $q$ and $P$ are closely related:
\begin{lem}
\label{lem:p_q}
Let $G$ be a graph. Then we have the polynomial identities
$q(G;x,y)=P(G;\frac{x-1}{y-1},y-1)$ and $P(G;u,x)=q(G;ux+1,x+1)$. \qed
\end{lem}

\subsection{Evaluating Graph Polynomials}
\label{ssec:graph_polynomials}

Given $\x, \y \in \Q$ we want to analyze the following computational
problem:
\begin{description}
\item[Input] Graph $G$
\item[Output] $q(G;\x,\y)$
\end{description}
This is what we mean by ``evaluating the interlace polynomial $q$ at
the point $(\x,\y)$''. As an abbreviation for this computational problem we write
\[ q(\x,\y), \] which should not be confused with the expression
$q(G;\x,\y)$ denoting just a value in $\Q$. Evaluating other graph
polynomials such as $P$, $q_N$, $q_R$ and $I$ is defined accordingly.

If $P_1$ and $P_2$ are computational problems we use $P_1 \red P_2$
($P_1 \redm P_2$) to denote a polynomial time Turing reduction
(polynomial time many-one reduction, resp.) from $P_1$ to $P_2$. For
instance, Lemma~\ref{lem:p_q} gives
\begin{cor}
  For $\x,\y\in \QQ$, $\y\neq 1$, we have $q(\x,\y)\redm
  P(\frac{\x-1}{\y-1},\y-1)$. For $\u,\x\in \QQ$ we have $P(\u,\x)\redm
  q(\u\x+1,\x+1)$. \qed
\end{cor}
Here $\QQ$ denotes some finite dimensional field extension $\Q\subs
\QQ \subs \R$, which has a discrete representation. 
As $\sqrt{2}$ will play an important role but we are
not able to use arbitrary real numbers as the input for a Turing
machine, we use $\QQ$ instead of $\Q$ or $\R$.  We fix some $\QQ$ for
the rest of this paper.
This construction is done in the spirit of Jaeger, Vertigan, and Welsh
\cite{jaeger_vertigan_welsh} who also propose to adjoin
a finite number of points to $\Q$ in order to talk  about
the complexity at irrational points. To some extent, this
is an ad hoc construction, but it is sufficient for this work.

\section{Graph Transformations for the Interlace Polynomial}
\label{sec:graph_trans}

Now we describe our graph transformations, vertex cloning and adding
combs or cycles to the vertices. The main results of this section are
Theorem~\ref{thm:p2p}, Theorem~\ref{thm:p2p_comb} and
Theorem~\ref{thm:p2p_cycles}, which describe the effect of these graph
transformations on the interlace polynomial.

\subsection{Cloning}
\label{ssec:cloning}

Cloning vertices in the graph yields our first graph transformation.

\subsubsection*{Cloning one vertex}

Let $G=(V,E)$ be a graph. Let $a\in V$ be some vertex (the one which
will be cloned) and $N$ the set of neighbors of $a$, $V'=V\setminus
\{a\}$ and $M=V'\setminus N$.  The graph $G$ with $a$ cloned,
$G_{aa}$, is obtained out of $G$ in the following way: Insert a new
isolated vertex $a'$.  Connect $a'$ to all vertices in $N$. If $a$
does not have a self loop, we are done.  Otherwise connect $a$ and
$a'$ and insert a self loop at $a'$. Thus, adjacency matrices of
the original (cloned, resp.)  graph are
\begin{equation}
B=
\begin{array}{c|ccc}
  & a    & N & M \\ \hline
a & b    & \vc{1} & \vc{0} \\
N & \vc{1}    & A_{11} & A_{12} \\
M & \vc{0}    & A_{21} & A_{22}
\end{array}
\quad\quad \text{and} \quad \quad
B_{aa} = \begin{array}{c|cccc}
   & a' & a    & N & M \\ \hline
a' & b  & b    & \vc{1} & \vc{0} \\
a  & b  & b    & \vc{1} & \vc{0} \\
N  & \vc{1}  & \vc{1}    & A_{11} & A_{12} \\
M  & \vc{0}  & \vc{0}    & A_{21} & A_{22}
\end{array}
\quad, \quad \text{resp,}
\end{equation}
where $b=1$ if $a$ has a self loop and $b=0$ otherwise. As the first
column of $B_{aa}$ equals its second column, as well as the first row
equals the second row, we can remove the first row and the first
column of $B_{aa}$ without changing the rank. This also holds when we
consider the adjacency matrices of $G[A]$ ($G_{aa}[A]$, resp.) instead
of $G$ ($G_{aa}$ resp.) for $A\subs V'$. Thus we have for any $A\subs
V'$
\begin{align}
\label{eq:rank_aa}
\rk(G_{aa}[A]) & = \rk(G[A]),\\
\label{eq:rank_aa2}
\rk(G_{aa}[A\uni\{a,a'\}]) & = \rk(\Ga[A\uni\{a\}]) =
\rk(\Ga[A\uni\{a'\}]) = \rk(G[A\uni\{a\}]).
\end{align}

Let $\gv{x} = (x_v)_{v\in V(G_{aa})}$ be a labeling of the vertices of
$G_{aa}$ by indeterminates. Define $\gv{X}$ to denote the following
labeling of the vertices of $G$: $X_v := x_v$ for all $v\in V'$, $X_a
:= (1+x_a)(1+x_{a'})-1 = x_a+x_{a'}+x_a x_{a'}$. Then we have

\begin{lem}
\label{lem:dupli}
$P(G_{aa};u,\gv{x}) = P(G;u,\gv{X})$.
\end{lem}

\begin{proof}
On the one hand we have
\begin{align*}
  & P(G_{aa};u,\gv{x}) \\ = & \sum_{A\subs V'} x_A (u^{\rk(\Ga[A])}
  + x_au^{\rk(\Ga[A\uni\{a\}])} + x_{a'}u^{\rk(\Ga[A\uni\{a'\}])} +
  x_a x_{a'}u^{\rk(\Ga[A\uni\{a, a'\}])}) \\
  =& \sum_{A\subs V'} x_A (u^{\rk(G[A])} + (x_a+x_{a'}+x_a
  x_{a'})u^{\rk(G[A\uni\{a\}])})\ \text{by \eqref{eq:rank_aa},
    \eqref{eq:rank_aa2}}.
\end{align*}

On the other hand we have
\begin{align*}
  P(G;u,\gv{X}) = & \sum_{A\subs V'}X_A(u^{\rk(G[A])} + X_a u^{\rk(G[A\uni \{a\}])})  \\
  =& \sum_{A\subs V'}x_A(u^{\rk(G[A])} + (x_a+x_{a'}+x_a x_{a'})
  u^{\rk(G[A\uni \{a\}])}).
\end{align*}
\end{proof}

\subsubsection*{Cloning all vertices}

Fix some $k$. Given a graph $G$, the graph $G_k$ is obtained by
cloning each vertex of $G$ exactly $k-1$ times. Note that the result
of the cloning is independent of the order in which the different
vertices are cloned. For $a\in V(G)$ let $a_1, \ldots, a_k$ be the
corresponding vertices in $G_k$.  For a vertex labeling $\gv{x}$ of
$G_k$ we define the vertex labeling $\gv{X}$ of $G$ by $X_a =
(1+x_{a_1})(1+x_{a_2})\cdots (1+x_{a_k})-1$ for $a\in V(G)$.  Applying
Lemma~\ref{lem:dupli} repeatedly we obtain

\begin{lem}
$P(G_k;u,\gv{x})=P(G;u,\gv{X})$. \qed
\end{lem}
Substitution of $x_v$ by $x$ for all vertices $v$ gives

\begin{thm}
\label{thm:p2p}
Let $G$ be a graph and $G_k$ be obtained out of $G$ by cloning each
vertex of $G$ exactly $k-1$ times. Then
\begin{equation}
\label{eq:P_clone}
P(G_k;u,x)=P(G;u, (1+x)^k-1).
\end{equation}
\qed
\end{thm}

As we will use it in the proof of Theorem~\ref{thm:q_summary}, we note
the following identity for $q$, which can be derived from
Theorem~\ref{thm:p2p} using Lemma~\ref{lem:p_q}:
\begin{align}
q(G_k;x,y)=q(G;(x-1)\frac{y^k-1}{y-1}+1, y^k).
\label{eq:Qtrans}
\end{align}

Theorem~\ref{thm:p2p} also implies the following reduction for the
interlace polynomial, which is the foundation for our results in
Section~\ref{sec:complexity}.

\newcommand{\Btwodef}{$B_2 = \{0, -1, -2\}$}
\begin{prop}
\label{prop:p_interpol}
\label{B2def}
Let \Btwodef\ and $x$ be an indeterminate. For $\u \in \QQ, \x\in
\QQ\setminus B_2$ we have $P(\u,x) \red P(\u,\x)$. (For any $\u\in
\QQ$, we write $P(\u,x)$ to denote the following computational
problem: given a graph $G$ compute $P(G;\u,x)$, which is a polynomial
in $x$ with coefficients in $\QQ$.)
\end{prop}
\begin{proof}
  Let $\u$ and $\x$ be given such that they fulfill the precondition
  of the proposition. Given a graph $G=:G_1$ with $n$ vertices, we
  build $G_2, G_3, \ldots, G_{n+1}$, where $G_i$ is obtained out of
  $G$ by cloning each vertex $i-1$ times. This is possible in time
  polynomial in $n$.  By Theorem~\ref{thm:p2p}, a call to an oracle
  for $P(\u,\x)$ with input $G_i$ gives us $P(G;\u,(1+\x)^i-1)$ for
  $i=1, \ldots, n+1$.  The restriction on $\x$ guarantees that for
  $i=1, 2, 3, \ldots$ the expression $(1+\x)^i-1$ evaluates to
  pairwise different values. Thus, for $P(G;\u,x)$, which is a
  polynomial in $x$ of degree $\leq n$, we have obtained the values at
  $n+1$ distinct points.  Using Lagrange interpolation we determine
  the coefficients of $P(G;\u,x)$.
\end{proof}

\subsection{Adding Combs}

The comb transformation sometimes helps, when cloning has not the
desired effect. Let $G=(V,E)$ be a graph and $a\in V$ some vertex.
Then we define the $k$-comb of $G$ at $a$ as $G_{a,k}=(V\uni \{a_1,
\ldots, a_k\}, E\uni \{ \{a,a_1\}, \ldots, \{a,a_k\} \})$, with $a_1,
\ldots, a_k$ being new vertices.

Using similar arguments as with vertex cloning, adding combs to
vertices yields a point-to-point reduction for the interlace
polynomial, too.
\begin{thm}
\label{thm:p2p_comb}
  Let $G$ be a graph and $G_k$ be obtained out of $G$ by performing a
  $k$-comb operation at every vertex. Then
\begin{equation}
\label{eq:combtrans}
P(G_k;u,x) =
p(k,u,x)^{|V(G)|}P(G;u,x/p(k,u,x)),
\end{equation}
where $p(k,u,x) = (1+x)^k(xu^2+1) - xu^2$.
\end{thm}
\begin{proof}
  The adjacency matrices of the original graph $G$ (the graph
  $G_{a,k}$ with a $k$-comb at $a$, resp.) are
\begin{equation}
  \label{eq:komb_matrix}
\begin{array}{c|cc}
  & a    & V' \\ \hline
a & b    & \vc{c} \\
V'& \vc{c}^T  & A
\end{array}
\quad\quad \text{and} \quad \quad
\begin{array}{c|cccccc}
      & a_1 & a_2 & \ldots & a_k & a & V' \\ \hline
  a_1 &     &     &  &     & 1 &       \\ 
  a_2 &     &     &  &     & 1 &        \\ 
  \vdots &  &   &   &  & \vdots &  \\ 
  a_k &     &     &  &     & 1 &        \\ 
  a   & 1   & 1   & \ldots &  1  & b    & \vc{c} \\
  V'  &        &        &  &  \vc{}   & \vc{c}^T  & A_{11} \\ 
\end{array}
\quad, \ \text{resp.},
\end{equation}
with empty entries being zero.  Consider $A\subs V(G_{a,k})$. Let $M
:= A \cut \{a, a_1, \ldots, a_k\}$. By \eqref{eq:komb_matrix}, the
rank of $G_{a,k}$ is related to the rank for $G$ in the following way:
\begin{itemize}
\item If $a \not \in M$, then $\rk(G_{a,k}[A]) = \rk(G[A\setminus M])$.
\item If $a \in M$ and $M\cut \{a_1, \ldots, a_k\}\neq \emptyset$,
  then $\rk(G_{a,k}[A])=\rk(G[A\setminus M])+2$: Let w.l.o.g.\ $a_1
  \in M$.  Consider the adjacency matrix of $G_{a,k}[A]$ and the
  following operations on it, which leave the rank unchanged. Using
  the first column we remove all $1$s in the $a$-row, except the $1$
  in the first column. Using the first row we remove all $1$s in the
  $a$-column, except the $1$ in the first row. The resulting matrix
  $B$ is a $(k+|V|) \times (k+|V|)$ matrix with $1$s at positions
  $(a,a_1)$ and $(a_1,a)$, the submatrix of $A_{11}$ induced by
  $A\setminus M$ in the lower right corner and zeros everywhere else.
  Thus $\rk(B)=\rk(G[A\setminus M])+2$.
\item If $M=\{a\}$, then $\rk(G_{a,k}[A]) = \rk(G[A])$.
\end{itemize}

Letting $r(A) := \rk(G[A])$ and $r_a(A) := \rk(G[A\uni \{a\}])$ for
$A\subs V'$, we see that $P(G_{a,k};u,\gv{x})$ equals
\begin{align*}
  \sum_{A\subs V'} x_A \Big(&u^{r(A)} \big(
  \underbrace{\sum_{\emptyset \subs S \subs \sak} x_S
    + x_a u^2 \cdot \sum_{\emptyset \subsetneq S \subs \sak} x_S}_{=:p(k,u,\gv{x})} \big) \\
  & + x_a u^{r_a(A)} \Big)
\end{align*}
Note that $p(k,u,\gv{x})$ does only depend on $x_a, x_{a_1}, \ldots,
x_{a_k}$, but not on $x_v$ for any $v\in V'$. As we have
\[ P(G;u,\gv{X}) = \sum_{A\subs V'}X_A (u^{r(A)} + X_a u^{r_a(A)}),
\]
we conclude
\[ P(G_{a,k};u,\gv{x}) = p(k,u,\gv{x})P(G;u,\gv{X}),
\]
where $X_v=x_v$ for $v\in V'$ and $X_a = \frac{x_a}{p(k,u,\gv{x})}$.

We can perform a $k$-comb operation at every $a\in V$ and call the
result $G_k$. Substituting $x$ for $x_v, v\in G_k,$ concludes the
proof.
\end{proof}

This yields
\begin{prop}
\label{prop:comb}
Let $p(k,u,x)=(1+x)^k(xu^2+1)-xu^2$. Let $k$ be a positive integer and
$\u,\x\in \QQ$. If $p(k,\u,\x)\neq 0$, we have
$P(\u,\x/p(k,\u,\x))\redm P(\u,\x)$. \qed
\end{prop}

\subsection{Adding Cycles}

Let $G=(V,G)$ be a graph and $a\in V$ some vertex. Consider the graph
$G_{a,k}=(V\uni\{1, 2, \ldots,
k-1\},E\uni\{\{a,1\},\{a,k-1\}\}\uni\{\{i-1,i\}\ |\ 1<i < k\})$, with
$1, 2, \ldots, k-1$ being new vertices. We say that $G_{a,k}$ has been
obtained out of $G$ by adding a $k$-cycle to $a$.

\begin{thm}
\label{thm:p2p_cycles}
  Let $G$ be a graph and $G_k$ be obtained out of $G$ by adding a
  $k$-cycle to every vertex. Then
  $P(G_k;u,x)=p_k(u,x)P(G;u,q_k(u,x)/p_k(u,x))$ for $k=3,4$ with
  $p_3(u,x)=1+2x+3x^2u^2$, $q_3(u,x)=x+x^3u^2$,
  $p_4(u,x)=1+3x+x^2+2x^2u^2+x^3u^2$ and
  $q_4(u,x)=x^2+2x^3u^2+x^4u^2$.
\end{thm}

\begin{prop}
\label{prop:cycles}
  $P(0,1)\redm P(0,-1)$ and $P(\u,-4) \redm P(\u,-2)$ for every $\u
  \in \QQ$.
\end{prop}
\begin{proof}
  The first reduction follows from Theorem~\ref{thm:p2p_cycles} adding
  $3$-cycles, the second adding $4$-cycles.
\end{proof}

\begin{proof}[Proof of Theorem~\ref{thm:p2p_cycles}]
  We use the same idea as in the proof of Theorem~\ref{thm:p2p_comb}.
  Consider the case of a $3$-cycle added at vertex $a$. Let
  $V'=V\setminus \{a\}$. The adjacency matrix of $G_{a,3}$ is
\[
\begin{array}{c|cccc}
    & 1 & 2 & a & V' \\ \hline
  1 &   & 1 & 1 &           \\ 
  2 & 1 &   & 1 &             \\ 
  a & 1 & 1 & b & \vc{c} \\
  V'&   &   & \vc{c}^T &  A_{11} \\ 
\end{array}
\]
with empty entries being zero. Adding the second row to the first row
and the second column to the first column and subsequently the first
row to the third row and the first column to the third column does not
change the rank and gives
\[
\begin{array}{c|cccc}
    & 1 & 2 & a & V' \\ \hline
  1 &   & 1 &   &           \\ 
  2 & 1 &   &   &             \\ 
  a &   &   & b & \vc{c} \\
  V'&   &   & \vc{c}^T &  A_{11} \\ 
\end{array}.
\]
This shows that $\rk(G_{a,3}[A]=\rk(G[A\setminus\{1,2\}])+2$ for all
$A$, $\{1,2,a\}\subs A \subs V(G_{a,3})$. Using arguments similar to
this one and the ones in the proof of Theorem~\ref{thm:p2p_comb} we
find that
\begin{itemize}
\item $\rk(G_{a,3}[A])=\rk(G[A\cut V']+2$ for all $A$, $\{a\}\subs A\subs
  V(G_{a,3})$ and either $1\in A$ or $2 \in A$,
\item $\rk(G_{a,3}[A])=\rk(G[A]$ for all $A$, $\{a\}\subs A \subs
  V(G_{a,3})$ and $\{1,2\}\cut A =\emptyset$,
\item $\rk(G_{a,3}[A])=\rk(G[A\cut V'])+\rk(P_2[A\cut V(P_2)])$ for
  all $A\subs V(G_{a,3})$, $a\not \in A$, where $P_2$ is the the path
  with two vertices $1,2$.
\end{itemize}
Letting again $r(A) := \rk(G[A])$ and $r_a(A) := \rk(G[A\uni \{a\}])$ for
$A\subs V'$, we see that $P(G_{a,3};u,\gv{x})$ equals
\begin{align*}
  \sum_{A\subs V'} x_A \Big(&u^{r(A)} \big(
  \underbrace{1+x_1+x_2+x_1x_2u^2+x_1x_au^2+x_2x_au^2}_{=:p_3(u,\gv{x})} \big) \\
  & + u^{r_a(A)}(\underbrace{x_a + x_1x_2x_au^2}_{=:q_3(u,\gv{x})}) \Big),
\end{align*}
which equals $p_3(u,\gv{x})P(G;u,\gv{X})$ if we define $\gv{X}$ by
$X_v=x_v$ for $v\in V'$ and $X_a=q_3(u,\gv{x})/p_3(u,\gv{x})$. We can
use this identity for every vertex $a$ and substitute $x_a$, $a\in V$,
by a single variable $x$. This gives the statement of the theorem
concerning $3$-cycles. For $4$-cycles we proceed in a similar fashion.
\end{proof}

\section{Complexity of evaluating the Interlace Polynomial exactly}
\label{sec:complexity}

The goal of this section is to uncover the complexity maps for $P$ and
$q$ as indicated in Figure~\ref{fig:pqcomplexity}. While the left hand
side (complexity map for $P$) is intended to follow the
\emph{arguments} which prove the hardness, the right hand side
(complexity map for $q$) focuses on presenting the \emph{results}.

\newcommand{\alphadef}{$\alpha = \sqrt{2}$}
\newcommand{\betadef}{$\beta = 1/\sqrt{2}$}

\begin{figure}
  \centering
  \includegraphics{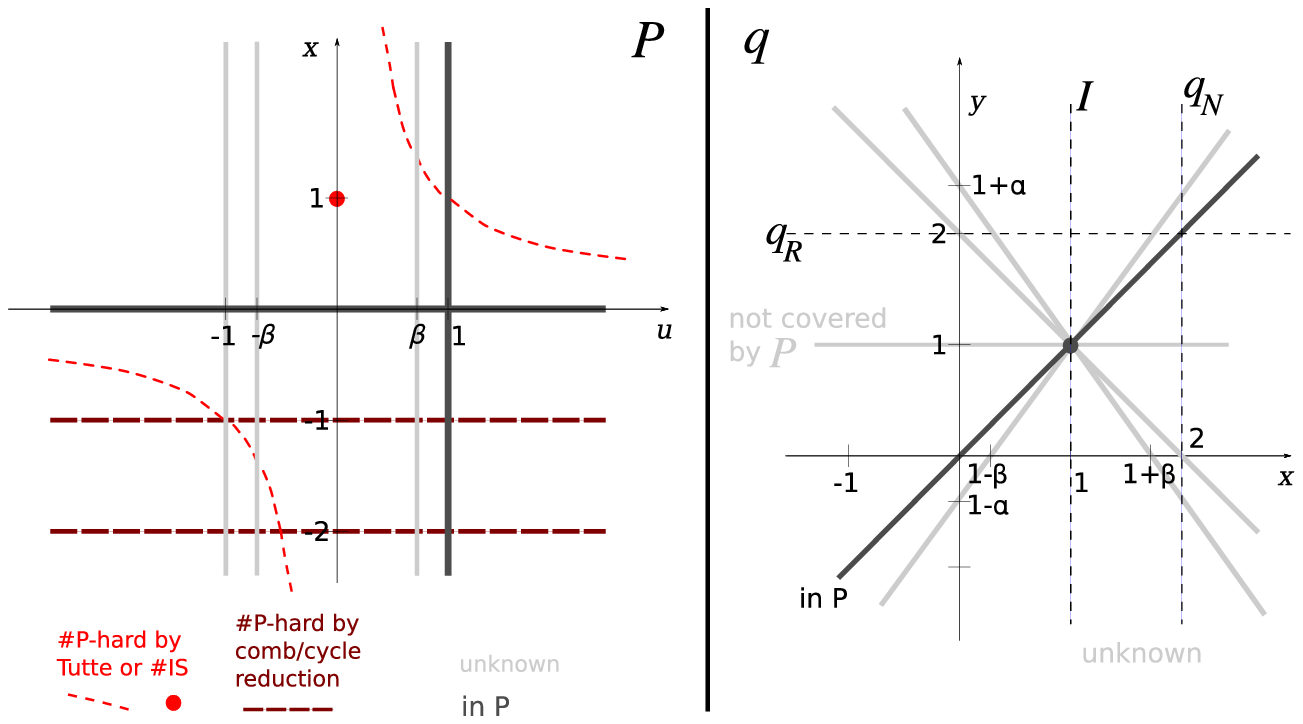}  
  \caption{Complexity of the interlace polynomials $P$ and $q$.
    \alphadef, \betadef}
  \label{fig:pqcomplexity}
\end{figure}

\begin{rem}
\label{rem:Peasy}
$P(\u,0)$ and $P(1,\x)$ are trivially solvable in polynomial time for
any $\u,\x\in \QQ$, as $P(G;\u,0)=1$ and $P(G;1,\x) = (1+\x)^{|V|}$.
\qed
\end{rem}
Thus, on the thick black lines $x=0$ and $u=1$ in the left half of
Figure~\ref{fig:pqcomplexity}, $P$ can be evaluated in polynomial
time. By Lemma~\ref{lem:p_q}, these lines in the complexity map for
$P$ correspond to the point $(1,1)$ and the line $x=y$, resp., in the
complexity map for $q$, see the right half of
Figure~\ref{fig:pqcomplexity}.

\subsection{Identifying hard points}
\label{ssec:collecting_hardness}

We want to establish Corollary~\ref{cor:qNhard} and
Remark~\ref{rem:IShard} which tell us, that $P$ is \SP-hard to
evaluate almost everywhere on the dashed hyperbola in
Figure~\ref{fig:pqcomplexity} and at $(0,1)$. To this end we collect
known hardness results about the interlace polynomial.

Let $t(G;x,y)$ denote the Tutte polynomial of an undirected graph
$G=(V,E)$. It may be defined by its state expansion as
\begin{equation} t(G;x,y)=\sum_{B\subs
\label{eq:t_def}
  E(G)}(x-1)^{r(E)-r(B)}(y-1)^{|B|-r(B)},
\end{equation}
where $r(B)$ is the $\F_2$-rank of the \emph{incidence} matrix of
$G[B] = (V,B)$, the subgraph of $G$ induced by $B$. (Note that $r(B)$
equals the number of vertices of $G[B]$ minus the number of components
of $G[B]$, which is the rank of $B$ in the cycle matriod of $G$.) For
details about the Tutte polynomial we refer to standard literature
\cite{tutte_graph_theory, tutte_appl, welsh_complexity}. The
complexity of the Tutte polynomial has been studied extensively.  In
particular, the following result is known.  \newcommand{\teasy}{\{0,
  1, 2, 1\pm \sqrt{2}\}}
\begin{thm}[\cite{vertigan_tutte_planar}]
\label{thm:tutteEasyXX}
Evaluating the Tutte polynomial of planar graphs at $(\x,\x)$ is
\SP-hard for all $\x\in \QQ$ except for $\x \in \teasy$.
\end{thm}
We will profit from this by a connection between the interlace
polynomial and the Tutte polynomial of planar graphs. This connection
is established via medial graphs. For any planar graph $G$ one can
build the oriented medial graph $\vec{G}_m$, find an Euler circuit $C$
in $\vec{G}_m$ and obtain the circle graph $H$ of $C$. The whole
procedure can be performed in polynomial time. For details we refer to
\cite{ems_distance_hereditary}. We will use

\begin{thm}[{\cite[End of Section 7]{interlace_polynomial}}; {\cite[Theorem 3.1]{ems_distance_hereditary}}]
\label{thm:tutte2interl}
Let $G$ be a planar graph, $\vec{G}_m$ be the oriented medial graph of
$G$ and $H$ be the circle graph of some Euler circuit $C$ of
$\vec{G}_m$.  Then $q(H; 2, y)=t(G;y,y)$. Thus we have $t(\y,\y)\redm
P(\frac{1}{\y-1},\y-1)$, where $t(\y,\y)$ denotes the problem of
evaluating the Tutte polynomial of a \emph{planar} graph at $(\y,\y)$.
\end{thm}
\begin{proof}
  See the references for $q(H;2,y)=t(G;y,y)$ and use
  Lemma~\ref{lem:p_q}.
\end{proof}

We set \alphadef\ and \betadef.
\newcommand{\Bonedef}{$B_1= \{\pm 1, \pm\beta, 0\}$}
\label{B1def}Let \Bonedef.  Theorem~\ref{thm:tutteEasyXX} and
Theorem~\ref{thm:tutte2interl} yield
\begin{cor}
\label{cor:qN_complexity}
\label{cor:qNhard}
Evaluating the vertex-nullity interlace polynomial $q_N$ is \SP-hard
almost everywhere. In particular, we have:
\begin{itemize}
\item The problem $q_N(2)$ is trivially solvable in polynomial time.
\item For any $\y\in \QQ \setminus \{0, 1, 2, 1\pm \alpha\}$ the
  problem $q_N(\y)=q(2,\y)$ is \SP-hard. Or, in other words, for any
  $\u\in \QQ\setminus B_1$ the problem $P(\u, 1/\u)$ is \SP-hard. \qed
\end{itemize}
\end{cor}
\begin{rem}
\label{rem:IShard}
$P(0,1)$ is \SP-hard, as $P(G;0,1)$ equals the number of independent sets of $G$, which is
  \SP-hard to compute \cite{counting_IS_hard}. \qed
\end{rem}

\subsection{Reducing to hard points}

The cloning reduction allows us to spread the collected hardness over
almost the whole plane: Combining Corollary~\ref{cor:qN_complexity}
and Remark~\ref{rem:IShard} with Proposition~\ref{prop:p_interpol} we
obtain
\begin{prop}
\label{prop:PHardness}
\label{prop:PHardness_u0}
Let \Bonedef\ and \Btwodef\ (as defined on Pages~\pageref{B1def} and
\pageref{B2def}, resp.). Let $(\u,\x)\in ((\QQ\setminus B_1)\uni
\{0\})\times (\QQ\setminus B_2)$.  Then $P(\u,\x)$ is \SP-hard.  \qed
\end{prop}
This tells us that $P$ is \SP-hard to evaluate at every point in the
left half of Figure~\ref{fig:pqcomplexity} not lying on one of the
seven thick lines (three of which are solid gray ones, two of which
are solid black ones, and two of which are dashed brown ones). Using
the comb and cycle reductions we are able to reveal the hardness of
the interlace polynomial $P$ on the lines $x=-1$ and $x=-2$:
\begin{prop}
\label{prop:PHardness_xm1}
  For $\u\in (\QQ\setminus B_1)\uni\{0\}$ the problem $P(\u,-1)$ is \SP-hard.
\end{prop}
\begin{proof}
  For $\u=0$ we use Proposition~\ref{prop:cycles} and
  Remark~\ref{rem:IShard}. If $\u\neq0$, we can use
  Proposition~\ref{prop:comb}, which yields $P(\u,-1/\u^2) \redm
  P(\u,-1)$. For $\u=\pm 1$ this reduces $(\pm 1,-1)$ to itself. For
  $\u=\pm \beta$ this reduces $(\beta, -2)$ to $(\beta,-1)$ and
  $(-\beta, -2)$ to $(-\beta,-1)$. For other $\u$ this gives a
  reduction of some point, which is already known as \SP-hard by
  Proposition~\ref{prop:PHardness}, to $(\u,-1)$.
\end{proof}

\begin{prop}
\label{prop:PHardness_xm2}
For $\u\in (\QQ\setminus B_1) \uni \{0\}$ the problem $P(\u,-2)$ is
\SP-hard.
\end{prop}

\begin{proof}
  Use Proposition~\ref{prop:cycles} and
  Proposition~\ref{prop:PHardness}.
\end{proof}

\subsection{Summing up}

First we summarize our knowledge about $P$.
\begin{thm}
\label{thm:p_summary}
Let \betadef.
\begin{enumerate}
\item $P(\u,\x)$ is computable in polynomial time on the lines $\u=1$
  and $\x=0$.
\item For $(\u,\x) \in (\QQ\setminus \{-1, -\beta, \beta, 1\})
  \times (\QQ\setminus \{0\})$ the problem
  $P(\u,\x)$ is \SP-hard.
\end{enumerate}
\end{thm}
\begin{proof}
  Summary of Remark~\ref{rem:Peasy}, Proposition~\ref{prop:PHardness},
  Proposition~\ref{prop:PHardness_xm1},
  Proposition~\ref{prop:PHardness_xm2}. The hardness of $P(0,-1)$
  follows from Corollary~\ref{cor:I_complexity}.
\end{proof}

In particular we obtain the following corollary about the complexity
of the independent set polynomial, which also follows from
\cite{averbouch_makowsky}.
\begin{cor}
\label{cor:I_complexity}
Evaluating the independent set polynomial
$I(\lambda)=P(0,\lambda)=q(1,1+\lambda)$ is \SP-hard at all $\lambda
\in \QQ$ except at $\lambda = 0$, where it is computable in polynomial
time.
\end{cor}

Now we turn to the complexity of $q$, see also the right half of
Figure~\ref{fig:pqcomplexity}.
\begin{thm}
\label{thm:q_summary}
  The two-variable interlace polynomial $q$ is \SP-hard to evaluate
  almost everywhere. In particular, we have:
\begin{enumerate}
\item \label{q_summary_polytime} $q(\x,\y)$ is computable in
  polynomial time on the line $\x=\y$.
\item \label{q_summary_y1} Let $\x\in \QQ\setminus \{1\}$ and $x$ be an
  indeterminate. Then $q(\x,1)$ is as hard as computing the whole
  polynomial $q(x,1)$.
\item \label{q_summary_SPhard} $q(\x,\y)$ is \SP-hard for all
  \begin{align*}(\x,\y) \in \{(\x,\y)\in \QQ^2\ |\ & \y \neq
    \pm(\x-1)+1 \ \text{and} \ \y\neq\pm \sqrt{2}(\x-1)+1\ \text{and}
    \ \y\neq 1 \}.
\end{align*}
\end{enumerate}
\end{thm}

\begin{proof}[Proof of Theorem~\ref{thm:q_summary} (Sketch)]
  (\ref{q_summary_polytime}) and (\ref{q_summary_SPhard}) follow from
  Remark~\ref{rem:Peasy} and Theorem~\ref{thm:p_summary} using
  Lemma~\ref{lem:p_q}. For $\x\neq 1$, \eqref{eq:Qtrans} gives
  $q(G_k;\x,1)=q(G;k(\x-1)+1,1)$, which yields enough points for
  interpolation in the same way as in
  Proposition~\ref{prop:p_interpol} using $k=1, 2, 3, \ldots$ This
  proves (\ref{q_summary_y1}).
\end{proof}
Theorem~\ref{thm:q_summary} implies
\begin{cor}
\label{cor:qR_complexity}
Let \betadef. Evaluating the vertex-rank interlace polynomial
$q_R(G;x)$ is \SP-hard at all $\x\in \QQ$ except at $\x=0, 1-\beta,
1+\beta$ (complexity open) and $\x=2$ (computable in polynomial time).
\qed
\end{cor}

\section{Inapproximability of the Independent Set Polynomial}
\label{sec:inapprox}

Provided we can evaluate the independent set polynomial at some fixed
point, vertex cloning (adding combs, resp.) allows us to evaluate it
at very large points. In this section we exploit this to prove that
the independent set polynomial is hard to approximate. Similar results
are shown in \cite{tutte_approx} for the Tutte polynomial.

\begin{defi}
  Let $\lambda \in \QQ$ and $\eps>0$. By a randomized
  $2^{n^{1-\eps}}$-approximation algorithm for $I(\lambda)$ we mean a
  randomized algorithm, that, given a graph $G$ with $n$ nodes, runs
  in time polynomial in $n$ and returns $\II(G;\lambda)\in \QQ$ such
  that
  \begin{equation*}
    \Pr[2^{-n^{1-\eps}}I(G;\lambda) \leq \II(G;\lambda) \leq
    2^{n^{1-\eps}}I(G;\lambda)] \geq \frac{3}{4}.
  \end{equation*}
\end{defi}

In \cite{tutte_approx}, (non)approximability in the weaker sense of
(not) admitting an \FPRAS\ is considered.
\begin{defi}
  Let $\lambda \in \QQ$. A fully polynomial randomized approximation
  scheme (\FPRAS) for $I(\lambda)$ is a randomized algorithm, that
  given a graph $G$ with $n$ nodes and an error tolerance $\eps, 0<
  \eps<1$, runs in time polynomial in $n$ and $1/\eps$ and returns
  $\II(G;\lambda) \in \QQ$ such that
  \begin{equation*}
    \Pr[2^{-\eps}I(G;\lambda) \leq \II(G;\lambda)\leq 2^{\eps}I(G;\lambda)]\geq \frac{3}{4}.
  \end{equation*}
\end{defi}

\begin{lem}
\label{lem:Iinapprox}
For every $\lambda \in \QQ$, $0\neq |1+\lambda|\neq 1$, and every
$\eps$, $0<\eps<1$, there is no randomized polynomial time
$2^{n^{1-\eps}}$-approximation algorithm for $I(\lambda)$ unless
$\RP=\NP$.
\end{lem}

\begin{thm}
\label{thm:I_inapprox}
For every $\lambda \in \QQ \setminus \{0\}$ and every $\eps$,
$0<\eps<1$, there is no randomized polynomial time
$2^{n^{1-\eps}}$-approximation algorithm (and thus also no \FPRAS) for
$I(\lambda)$ unless $\RP=\NP$.
\end{thm}
\begin{proof}
  Lemma~\ref{lem:Iinapprox} gives the inapproximability at $\lambda
  \in \QQ\setminus \{-2,-1,0\}$. By \eqref{eq:combtrans} we could turn
  an approximation algorithm for $I(-2)$ into an approximation
  algorithm for $I(2)$ which would imply $\RP=\NP$ by
  Lemma~\ref{lem:Iinapprox}. For $I(-1)$ we use
  Theorem~\ref{thm:p2p_cycles}.
\end{proof}

\begin{proof}[Proof of Lemma~\ref{lem:Iinapprox}]
  Fix $\lambda\in \QQ, 0\neq |1+\lambda|\neq1$, and $\eps$,
  $0<\eps<1$.  Assume we have a randomized
  $2^{n^{1-\eps}}$-approximation algorithm $\aA$ for $I(\lambda)$.
  Given a graph $G$, Theorem~\ref{thm:p2p} and
  Theorem~\ref{thm:p2p_comb}, resp.,\ will allow us to evaluate the
  independent set polynomial at a point $\x$ with $|\x|$ that large,
  that an approximation of $I(G;\x)$ can be used to recover the degree
  of $I(G;x)$, which is the size of a maximum independent set of $G$.
  As computing this number is \NP-hard, a randomized
  $2^{n^{1-\eps}}$-approximation algorithm for $I(G;\lambda)$ would
  yield an \RP-algorithm for an \NP-hard problem, which implies
  $\RP=\NP$.

  Let $G=(V,E)$ be a graph with $|V|=n$. We distinguish two cases. If
  $|1+\lambda|>1$, we choose a positive integer $l$ such that
  $(nl)^{1-\eps}\geq n^2$ and with $\x:=(1+\lambda)^l-1$ we have
  \begin{equation}
    \label{eq:large_point}
    |\x|>  2^{2(nl)^{1-\eps}+n+2}.
  \end{equation}
  As $\lambda$ and $\eps$ are constant, this can be achieved by
  choosing $l=\poly(n)$. If $0<|1+\lambda|<1$, we choose a positive
  integer $l$ such that with $\x:=\frac{\lambda}{(1+\lambda)^l}$
  \eqref{eq:large_point} holds.  By Theorem~\ref{thm:p2p}
  (Theorem~\ref{thm:p2p_comb}, resp.)  we have
  $I(G;\x)=I(G_l;\lambda)$ ($I(G;\x)=(1+\lambda)^{-l|V|}
  I(G_l;\lambda)$, resp.). Algorithm $\aA$ returns on input $G_l$
  within time $\poly(nl)=\poly(n)$ an approximation
  $\II(G_l;\lambda)$, such that with $\II(G;\x):=\II(G_l;\lambda)$
  ($\II(G;\x):=\frac{\II(G_l;\lambda)}{(1+\lambda)^{l|V|}}$, resp.) we
  have
  \begin{equation}
    \label{eq:approx}
    2^{-(nl)^{1-\eps}}I(G;\x)\leq \II(G;\x)\leq 2^{(nl)^{1-\eps}} I(G;\x)
  \end{equation}
  with high probability.

  Let $c$ be the size of a maximum independent set of $G$, and let $N$
  be the number of independent sets of maximum size. We have
  \begin{align*}
    I(G;x) = Nx^c + \sum_{0\leq j\leq c-1}i(G;j)x^j
  \end{align*}
  and thus
  \begin{equation}
  \begin{split}
    \Big|\frac{I(G;\x)}{\x^c} - N\Big| &\leq \sum_{0\leq j\leq c-1} i(G;j)|\x|^{j-c} \\
    & \leq c 2^n|\x|^{-1} \leq 2^{\log n +n} |\x|^{-1}\overset{\eqref{eq:large_point}}{<}\frac{1}{2}.
    \label{eq:N_estimate}
  \end{split}
  \end{equation}
  If we could evaluate $I(G;\x)$ exactly, we could try all $c\in \{1,
  \ldots, n\}$ to find the one for which $\frac{I(G;\x)}{\x^c}$ is a
  good estimation for $N$, $1\leq N\leq 2^n$. This $c$ is unique as
  $|\x|>2^{n^2}$. The following calculation shows that this is also
  possible using the approximation algorithm~$\aA$.

  Using $\aA$ we compute $\NN(\cc) := \frac{\II(G;\x)}{\x^{\cc}}$ for
  all $\cc \in \{1, \ldots, n\}$. We claim that $c$ is the unique
  $\cc$ with
  \begin{equation}
    \label{eq:NN_bounds}
     2^{-(nl)^{1-\eps}-1} \leq \NN(\cc) \leq 2^{(nl)^{1-\eps}+n+1}.
  \end{equation}
  Let us prove this claim. As $1\leq N\leq 2^n$ and by
  \eqref{eq:N_estimate}, we know that
  \begin{equation}
    \label{eq:fraction_bounds}
    \frac{1}{2} \leq \frac{I(G,\x)}{\x^c}\leq 2^{n+1}.
  \end{equation}
  Thus, by \eqref{eq:approx}, $\cc=c$ fulfills \eqref{eq:NN_bounds}.

  On the other hand, when $\cc\leq c-1$ we have
  \begin{equation*}
    |\NN(\cc)| \overset{\eqref{eq:approx}, \eqref{eq:fraction_bounds}}{\geq}
    2^{-(nl)^{1-\eps} -1}|\x| \overset{\eqref{eq:large_point}}{>} 2^{(nl)^{1-\eps}+n+1}.
  \end{equation*}
  When $\cc\geq c+1$ we have $|\NN(\cc)| < 2^{-(nl)^{1-\eps}-1}$ by similar arguments.
  This shows that any integer $\cc, \cc\neq c,$ does
  not fulfill \eqref{eq:NN_bounds}. Thus, $c$ can be found in
  randomized polynomial time using $\aA$.
\end{proof}

\section*{Acknowledgement}
We would like to thank Johann A.\ Makowsky for valuable comments on an
earlier version of this work and for drawing our attention to
\cite{averbouch_makowsky}.

\bibliographystyle{alpha}
\bibliography{literatur}

\begin{thebibliography}{ABCS00}

\bibitem[ABCS00]{dna_interlacings}
Richard Arratia, B\'{e}la Bollob\'{a}s, Don Coppersmith, and Gregory~B. Sorkin.
\newblock Euler circuits and {DNA} sequencing by hybridization.
\newblock {\em Discrete Applied Mathematics}, 104(1-3):63--96, 15 August 2000.

\bibitem[ABS04a]{interlace_polynomial}
Richard Arratia, B\'{e}la Bollob\'{a}s, and Gregory~B. Sorkin.
\newblock The interlace polynomial of a graph.
\newblock {\em J. Comb. Theory Ser. B}, 92(2):199--233, 2004.

\bibitem[ABS04b]{arratia_two_var_interl}
Richard Arratia, B\'{e}la Bollob\'{a}s, and Gregory~B. Sorkin.
\newblock A two-variable interlace polynomial.
\newblock {\em Combinatorica}, 24(4):567--584, 2004.

\bibitem[AM07]{averbouch_makowsky}
Ilia Averbouch and J.~A. Makowsky.
\newblock The complexity of multivariate matching polynomials, February 2007.
\newblock Preprint.

\bibitem[BM06]{blaes_mak}
Markus Bl\"aser and Johann Makowsky.
\newblock Hip hip hooray for {S}okal, 2006.
\newblock Unpublished note.

\bibitem[BO92]{tutte_appl}
Thomas Brylawski and James Oxley.
\newblock The {Tutte} polynomial and its applications.
\newblock In Neil White, editor, {\em Matroid Applications}, Encyclopedia of
  Mathematics and its Applications. Cambridge University Press, 1992.

\bibitem[Cou07]{courcelle_interl}
Bruno Courcelle.
\newblock A multivariate interlace polynomial, 2007.
\newblock Preprint, arXiv:cs.LO/0702016 v2.

\bibitem[EMS07]{ems_distance_hereditary}
Joanna~A. Ellis-Monaghan and Irasema Sarmiento.
\newblock Distance hereditary graphs and the interlace polynomial.
\newblock {\em Comb. Probab. Comput.}, 16(6):947--973, 2007.

\bibitem[GJ07]{tutte_approx}
Leslie~Ann Goldberg and Mark Jerrum.
\newblock Inapproximability of the {T}utte polynomial.
\newblock In {\em STOC '07: Proceedings of the 39th Annual ACM Symposium on
  Theory of Computing}, pages 459--468, New York, NY, USA, 2007. ACM Press.

\bibitem[JVW90]{jaeger_vertigan_welsh}
F.~Jaeger, D.~L. Vertigan, and D.~J.~A. Welsh.
\newblock On the computational complexity of the {J}ones and the {T}utte
  polynomials.
\newblock {\em Math. Proc. Cambridge Philos. Soc.}, 108:35--53, 1990.

\bibitem[Sok05]{sokal-2005}
Alan~D. Sokal.
\newblock The multivariate {T}utte polynomial (alias {P}otts model) for graphs
  and matroids.
\newblock In Bridget~S. Webb, editor, {\em Surveys in Combinatorics 2005}.
  Cambridge University Press, 2005.

\bibitem[Tut84]{tutte_graph_theory}
W.~T. Tutte.
\newblock {\em Graph {Theory}}.
\newblock Addison Wesley, 1984.

\bibitem[Val79]{counting_IS_hard}
Leslie~G. Valiant.
\newblock The complexity of enumeration and reliability problems.
\newblock {\em SIAM Journal on Computing}, 8(3):410--421, 1979.

\bibitem[Ver05]{vertigan_tutte_planar}
Dirk Vertigan.
\newblock The computational complexity of {T}utte invariants for planar graphs.
\newblock {\em SIAM Journal on Computing}, 35(3):690--712, 2005.

\bibitem[Wel93]{welsh_complexity}
D.~J.~A. Welsh.
\newblock {\em Complexity: knots, colourings and counting}.
\newblock Cambridge University Press, New York, NY, USA, 1993.

\end{thebibliography}

\end{document}